\documentclass[final]{amsart}
\usepackage{amsmath}
  \usepackage{paralist}
 \usepackage[colorlinks=true]{hyperref}
\hypersetup{urlcolor=blue, citecolor=red}
\usepackage[notref,notcite]{showkeys}

\def\arXiv#1{
  {\href{http://arxiv.org/pdf/#1}
   {{\mdseries\ttfamily arXiv:#1}}}}

\usepackage{amssymb}
\usepackage{amstext}
\usepackage{amsthm}
\usepackage{amsfonts}
\usepackage{amsxtra}
\usepackage{graphicx}

\newcommand{\R}{\mathbb R}

\newcommand{\de}{\partial}
\newcommand{\imp}{\mathrm{Im}\,}

\newcommand{\vc}[1]{\boldsymbol #1}

\newcommand{\lap}{\Delta}

\newcommand{\dvg}{\operatorname{\mathrm{div}}}
\newcommand{\sym}{\operatorname{\mathrm{Sym}}}

\newcommand{\vt}[1]{\mathsf #1}

\newcommand{\bs}{\Sigma}
\newcommand{\rey}{\mathit{Re}}

\newcommand{\tsp}{\intercal}

\newcommand{\slet}{\boldsymbol{\zeta}}

\newcommand{\sletp}{p_{\boldsymbol{\zeta}}}
\newcommand{\tslet}{\vt Z}
\newcommand{\tsletc}{Z}
\newcommand{\ft}[1]{\widehat{#1}}

\newtheorem{thm}{Theorem}[section]
\newtheorem{lem}[thm]{Lemma}
\newtheorem{cor}[thm]{Corollary}
\newtheorem{prop}[thm]{Proposition}

\theoremstyle{definition}
\newtheorem{defn}[thm]{Definition}
\theoremstyle{definition}

\theoremstyle{remark}

\title[Free fall of one-dimensional bodies] 
{Steady free fall of one-dimensional bodies in a
  hyperviscous fluid at low Reynolds number}

\author[G. G. Giusteri, A. Marzocchi and A. Musesti]{}


 \email{giulio.giusteri@unicatt.it}
 \email{alfredo.marzocchi@unicatt.it}
 \email{alessandro.musesti@unicatt.it}

\begin{document}

\maketitle

\centerline{\scshape Giulio G. Giusteri, Alfredo Marzocchi and
  Alessandro Musesti }
\medskip
{\footnotesize
 \centerline{Dipartimento di Matematica e Fisica ``N. Tartaglia''}
   \centerline{Universit\`a Cattolica del Sacro Cuore}
   \centerline{Via dei Musei 41, I-25121 Brescia, Italy}
}
\bigskip


\begin{abstract}
  The paper is devoted to the study of the motion of one-dimensional
  rigid bodies during a free fall in a quasi-Newtonian hyperviscous
  fluid at low Reynolds number. We show the existence of a steady
  solution and furnish sufficient conditions on the geometry of the
  body in order to get purely translational motions. Such conditions
  are based on a generalized version of the so-called {\em Reciprocal
    Theorem} for fluids.
\end{abstract}

\section{Introduction}

The study of the free fall of slender bodies in liquids is an essential
issue in many problems of practical interest, such as the design of
composite materials or the analytical technique of separation of
macromolecules by electrophoresis (see \cite{Gal02} for a very
interesting and rich review on the subject). Typical experiments show
that homogeneous bodies satisfying some symmetry conditions, when
dropped in a quiescent viscous liquid, will eventually reach a steady
state that is purely translational, having the symmetry axis
forming an angle with respect to the gravity $\vc g$, called {\em tilt
  angle}, that depends on the material geometry of the body and on the
physical properties of the liquid.

If the geometry of the body is such that one of the dimensions
dramatically prevails on the other two, the assumption that the body
is one-dimensional is a reasonable simplification which can give
satisfactory results. However, a one-dimensional body is ``too thin''
to interact with a classical Newtonian incompressible fluid in 3D (it
has null capacity, see~\cite{Wei72}). Hence, we propose to study the
problem of the free fall of a slender body in a regularized model for
Newtonian fluids, introduced by Fried and Gurtin ~\cite{FriGur06} in
2006, where higher-order derivatives are considered in the
constitutive prescription of the Cauchy stress tensor.

The Navier--Stokes equation for incompressible fluids reads
\begin{equation*}
\rho\frac{\de\vc u}{\de t}+\rho(\vc u\cdot\nabla)\vc u+\nabla p-\mu\lap\vc u=\rho\vc b,
\end{equation*}
where $p$ is the pressure field, $\vc u$ is the divergence-free
velocity field, $\rho>0$ is the constant and homogeneous mass density,
$\mu>0$ is the dynamic viscosity and $\rho\vc b$ is a volumetric force
density. The {\em hyperviscous regularization} consists of adding a
term proportional to $\lap\lap\vc u$ to the equation. For this
modified equation,
\begin{equation*}
\rho\frac{\de\vc u}{\de t}+\rho(\vc u\cdot\nabla)\vc u+\nabla p-\mu\lap\vc u
+\zeta\lap\lap\vc u=\rho\vc b,
\end{equation*}
where $\zeta>0$ is the {\em hyperviscosity}, the existence and uniqueness of
regular solutions have been established. In a series of papers
\cite{FriGur06,Mus09,GiuMar11,Giu13} a purely mechanical explanation
of the hyperviscous term has been proposed, and different
contributions to $\zeta$ associated with dissipation functionals are
introduced and analyzed.  Here we assign to $\zeta$ a geometric role, by
introducing the \emph{effective thickness} $L>0$ of the
lower-dimensional objects, and setting $\zeta=\mu L^2$, so that the
hyperviscous flow equation becomes
\begin{equation*}
\rho\frac{\de\vc u}{\de t}+\rho(\vc u\cdot\nabla)\vc u+\nabla p-\mu\lap\left(\vc u
-L^2\lap\vc u\right)=\rho\vc b.
\end{equation*}
In the experiments on the free fall of rigid bodies in viscous fluids
the Reynolds number is often very small, so that the inertia of the
liquid can be neglected and one can linearize the flow equation
\cite{Wei72,Bre64}. However, even after that approximation the problem
does not become fully linear, since there remains a nonlinear coupling
between the flow and the rigid body motion.

In the present paper we will study the steady fall of a one-dimensional
rigid body in a hyperviscous fluid at low Reynolds number (for a
treatment of the full nonlinear unsteady problem we refer the reader
to~\cite{GiuMar13}). In Section~\ref{sec:formulation} we give the
mathematical formulation of the problem, and in Section~\ref{sec:drhr}
we study the forces acting on the body in the case of a hyperviscous
fluid. In Section~\ref{sec:freefall} we show the existence of steady
solutions, and in Section~\ref{sec:res} we prove the Reciprocal
Theorem in the case of a linearized hyperviscous liquid surrounding a
one-dimensional body, and study some properties of the so-called {\em
  resistance tensors}. Finally, Section~\ref{sec:sym} contains
sufficient conditions on the geometry of a homogeneous body in order
to get purely translational solutions.

\section{Formulation of the free fall problem}\label{sec:formulation}

The free fall problem is characterized by the fact that the rigid body
is immersed and dropped from rest in an otherwise quiescent fluid and
gravity is the only external force acting on the system.  We represent
a rigid body as a connected, bounded, closed subset $\bs$ of $\R^3$
which is a finite union of images of $[0,1]$ through 
$C^1$-diffeomorphisms.  It is convenient to write the problem in a
co-moving frame centered at the center of mass $\vc c(t)$ of
$\bs$.  Denoting by $\vc y$ the position of a point in the original
inertial frame, and by $\vc x$ its position in the co-moving frame, we
know that, at any time $t\geq 0$,
\[
\vc x=\vt Q(t)(\vc y-\vc c(t)),
\]
where $\vt Q(t)$ is an orthogonal linear transformation for any $t\geq 0$, with $\vt Q(0)=\vt 1$.
If the velocity of the center of mass and the spin of the rigid body in the inertial frame are denoted by $\vc\eta(t)$ and $\vc\varOmega(t)$, respectively, so that
\[
\vc v(t)=\vc\eta(t)+\vc\varOmega(t)\times(\bar{\vc y}-\vc c(t))
\]
is the velocity, in that frame, of any point $\bar{\vc y}$ belonging to the rigid body, then their expression in the co-moving frame is given by
\[
\vc\xi(t):=\vt Q(t)\vc\eta(t)\quad\text{and}\quad\vc\omega(t):=\vt Q(t)\vc\varOmega(t),
\]
respectively, and the rigid velocity field $\vc v$ is transformed into
\[
\vc U(t):=\vc\xi(t)+\vc\omega(t)\times\bar{\vc x}, 
\]
where $\bar{\vc x}$ denotes the coordinates of $\bar{\vc y}$ in the co-moving frame.

The gravitational acceleration vector $\vc g$ (constant in the
inertial frame) is represented, in the co-moving frame, by $\vc
G(t):=\vt Q(t)\vc g$, which is easily seen to satisfy the
ordinary differential equation
\begin{equation}\label{eq:G}
\frac{d\vc G}{dt}=\vc G\times\vc\omega.
\end{equation}

As customary when studying flows past rigid bodies, the velocity field
$\vc u$ that we consider is the so called \emph{disturbance field},
which is the difference between the actual flow and the flow at
infinity, both seen in the co-moving frame. Since the flow at infinity
is $-\vc U$ (that is minus the extension to all of the fluid of the
motion of the immersed object), the representation of the fluid flow
in the co-moving frame is given by $\vc u-\vc U$.

The continuity and flow equations for an incompressible (disturbance)
velocity field $\vc u(\vc x,t)$ and pressure field $p(\vc x,t)$
defined on $(\R^3\setminus\bs)\times[0,+\infty)$, become then
\begin{gather}\label{eq:incomp}
\dvg\vc u=0,\\
\rho\left(\frac{\de\vc u}{\de t}+[(\vc u-\vc U)\cdot\nabla]\vc u+\vc\omega\times\vc u\right)=\dvg\vt T(\vc u,p)+\rho\vc G,
\end{gather}
where $\vt T(\vc u,p)$ denotes the Cauchy stress tensor.  Notice that,
thanks to its frame indifference properties, $\vt T$ retains the same
functional dependence on the velocity field seen both in the inertial
frame and in the co-moving one.

The disturbance field $\vc u$ satisfies also the decay condition
\begin{equation}\label{eq:decay}
\lim_{|\vc x|\to\infty}\vc u(\vc x,t)=\vc 0,
\end{equation}
and the adherence to the rigid body, given by
\begin{equation}\label{eq:adherence}
\vc u(\vc x,t)=\vc U(\vc x,t)\qquad\text{on }\bs\times[0,+\infty).
\end{equation}

To properly account for Archimedean forces we introduce the {\em
  effective mass} of the body as given by
\[
m_e=m-m_c,
\]
that is, the difference between the real mass of the object and the
{\em complementary mass} $m_c$ of a portion of fluid occupying the
real volume of the object. Even if for a one-dimensional body the
complementary mass $m_c$ should vanish, in view of the interpretation
of such bodies as representations for real
three-dimensional objects we allow for any value $0\leq m_c\leq m$,
suggested by the physical properties of the interaction between the
body and the fluid. Then the equations of motion for the rigid body in
the co-moving frame read
\begin{gather}\label{eq:linmom}
m\frac{d\vc\xi}{dt}+m\vc\omega\times\vc\xi=m_e\vc G+\vc f(\vc u,p),\\
\label{eq:angmom}
\vt J\frac{d\vc\omega}{dt}+\vc\omega\times(\vt J\vc\omega)=
-m_c\vc r\times\vc G+\vc t(\vc u,p),
\end{gather}
where $\vt J$ is the inertia tensor of $\bs$, $\vc r$ is the position
of the centroid\footnote{Notice that the centroid of a rigid body
  coincides with its center of mass when the body has a uniform mass
  density; in the latter case $\vc r=\vc 0$.} in the co-moving frame,
and $\vc f$ and $\vc t$ are the total hydrodynamic force and torque,
respectively, exerted on the rigid body as a consequence of the fluid
flow $\vc u$ and pressure $p$. The proper definitions of $\vc f$ and
$\vc t$, as well as the physical properties of the fluid
encoded in the Cauchy stress tensor $\vt T$, are discussed in
Section~\ref{sec:drhr}.

The whole set of equations \eqref{eq:G}--\eqref{eq:angmom} represents
the differential problem associated with the free fall of a rigid
object $\bs$ in an incompressible fluid. It is convenient to consider
it in the non-dimensional form: by choosing suitable reference length
$d$, proportional to the diameter of $\bs$, and speed $W=\rho
  g d^2/\mu$, we can switch to non-dimensional quantities
according to
\begin{gather*}
\vc x\to\frac{\vc x}{d}\;,\quad t\to\frac{t\mu}{\rho d^2}\;,\quad\vc
u\to\frac{\vc u}{W}\;,\quad \vc \xi\to\frac{\vc \xi}{W}\;,\quad
\vc\omega\to\frac{\vc \omega d}{W}\;,\\
p\to\frac{pd}{\mu W}\;,\quad
m\to\frac{m}{\rho d^3}\;,\quad\vc G\to \frac{\vc G}{g}\;,\quad\vc g\to \frac{\vc g}{g}\;,
\end{gather*}
obtaining
\begin{align}
&\dvg\vc u=0,\label{eq:nd1}\\
&\frac{\de\vc u}{\de t}+\rey\left\{[(\vc u-\vc U)\cdot\nabla]\vc u+\vc\omega\times\vc u\right\}=\dvg\vt T(\vc u,p)+\vc G,\label{eq:nd2}\\
&\lim_{|\vc x|\to\infty}\vc u(\vc x,t)=\vc 0,\\
&\vc u(\vc x,t)=\vc U(\vc x,t)\qquad\text{on }\bs\times[0,+\infty),\label{eq:nd4}\\
&m\frac{d\vc\xi}{dt}+\rey(m\vc\omega\times\vc\xi)=
m_e\vc G+\vc f(\vc u,p),\label{eq:nd5}\\
&\vt J\frac{d\vc\omega}{dt}+\rey[\vc\omega\times(\vt J\vc\omega)]=
-m_c\vc r\times\vc G+\vc t(\vc u,p),\label{eq:nd6}\\
&\frac{d\vc G}{dt}=\rey(\vc G\times\vc\omega),\label{eq:nd7}
\end{align}
with initial conditions
\begin{equation}\label{eq:ic}
\vc u(\vc x,0)=\vc\xi(0)=\vc\omega(0)=\vc 0,\quad\vc G(0)=\vc g,
\end{equation}
where $\rey=\rho Wd/\mu=\rho^2gd^3/\mu^2$ is the Reynolds number and every quantity has
to be understood as non-dimensional.

The low-Reynolds-number approximation of the differential problem,
which is also a linearization of the equation for the flow, is
obtained by neglecting the terms proportional to $\rey$ in
equations~\eqref{eq:nd2}, \eqref{eq:nd5}, and \eqref{eq:nd6}. When
considering free fall problems, the energy budget is determined by
gravitational forces and viscous dissipation; hence the limit $\rey\to
0$ corresponds to the situation where the latter prevails. In the
meanwhile, the geometric parameters $d$ and $L$ do not need to be small,
even in that limit.

Notice that equation \eqref{eq:nd7} remains unchanged, since it
represents a geometric constraint which holds for any non-vanishing
value of $\rey$. Moreover, the steady version of the problem is
achieved by assuming that all the quantities do not depend on time,
hence neglecting all the time derivatives. In that case, keeping into
account also the initial conditions~\eqref{eq:ic},
equation~\eqref{eq:nd2} becomes
\begin{equation}\label{eq:motion}
\dvg\vt T(\vc u,p)+\vc g=0\quad\text{on $\R^3\setminus\bs$},
\end{equation}
and equation~\eqref{eq:nd4} writes
\[
\vc u(\vc x)=\vc U(\vc x) = \vc\xi+\vc\omega\times\vc x\quad\text{on
}\bs.
\]
Finally, equations~\eqref{eq:nd5}--\eqref{eq:nd7} become
\[
m_e\vc g+\vc f(\vc u,p)=0,\quad -m_c\vc r\times\vc g+\vc t(\vc u,p)=0,
\quad \vc g\times\vc\omega=\vc 0,
\]
respectively.

\section{The viscous force acting on a slender
  body}\label{sec:drhr}

Given $r>0$, we introduce the (closed) $r$-neighborhood of the slender
body $\bs$ by setting
\[
V_r(\bs):=\left\{\vc x\in\R^3 : \inf_{\vc c\in \bs}|\vc x - \vc c|\leq r \right\}.
\]
Then we define the total hydrodynamic force, due to the fluid velocity
and pressure field $(\vc u,p)$, acting on $\bs$ as
\begin{equation}\label{eq:defforce}
\vc f(\vc u,p):=\lim_{r\to 0}\int_{\de V_r(\bs)}\vt T(\vc u,p)\vc n,
\end{equation}
where $\vc n$ denotes the unit outer normal to $\de V_r(\bs)$. Notice
that, thanks to the regularity of $\bs$, the $r$-neighborhood
$V_r(\bs)$ has a Lipschitz boundary for any $r$ sufficiently small.
\begin{prop}
The limit in~\eqref{eq:defforce} is well-defined.
\end{prop}
\begin{proof}
  We consider a ball $B_R$ centered at the origin and with radius $R$,
  which contains $V_r(\bs)$ for some $r>0$. According to
  equation~\eqref{eq:motion}, the term $\dvg\vt T$ balances the
  gravity, so that it is represented by a measure whose singular part
  is concentrated on $\bs$ (we can see this by noting that the mass
  density per unit volume must diverge on $\bs$ to give a non-zero
  weight to a body with vanishing volume). Denote by $\lambda_{ac}$ and
  $\lambda_s$ the absolutely continuous and singular parts of the
  measure $\dvg\vt T$, respectively. It follows that the support of
  $\lambda_s$ is contained in $\bs$ and that $\dvg\vt T=\lambda_{ac}$ in
  $\R^3\setminus\bs$. Then we have, by applying
  Lebesgue's theorem,
\begin{equation}\label{eq:compl}
\lim_{r\to 0}\int_{B_R\setminus V_r(\bs)}\dvg\vt T=\lim_{r\to 0}\int_{B_R\setminus V_r(\bs)}\lambda_{ac}=\int_{B_R\setminus\bs}\lambda_{ac}=\int_{B_R\setminus\bs}\dvg\vt T,
\end{equation}
since $\lambda_{ac}\in L^1(\R^3;\R^3)$. Then, by the Divergence theorem,
\begin{equation}\label{eq:limf}
\vc f=\lim_{r\to 0}\int_{\de V_r(\bs)}\vt T\vc n=\int_{\de B_R}\vt T\vc n-\int_{B_R\setminus\bs}\dvg\vt T,
\end{equation}
where $\vc n$ is always the outer normal. Since the right-hand side
of~\eqref{eq:limf} is independent of $r$, the left-hand side is well
defined. 
\end{proof}
In a similar fashion, we define the total hydrodynamic torque acting
on $\bs$, due to the fluid velocity and pressure field $(\vc u,p)$, as
\begin{equation*}
\vc t(\vc u,p):=\lim_{r\to 0}\int_{\de V_r(\bs)}\vc x\times\vt T(\vc u,p)\vc n.
\end{equation*}
It is important to stress the fact that, if $\lambda_s$ were absent,
{\em i.e.}\ if $\dvg \vt T$ were an $L^1$-function, then the integral
over $B_R\setminus\bs$ of $\dvg\vt T$ in~\eqref{eq:limf} would be
equal to its integral over all of $B_R$ and $\vc f$ and $\vc t$ would
simply vanish. 

\null

The constitutive theory for non-simple fluids leading to a
hyperviscous flow equation has been developed
in~\cite{FriGur06,Mus09,GiuMar11}. It offers a number of possible
choices for the terms to be included in $\vt T$, in addition to those
of Newtonian fluids. Here we make the following, somewhat minimal, choice:
\begin{equation}\label{eq:defT}
\vt T(\vc u,p):=-p\vt 1+\left(\nabla\vc u+\nabla\vc u^\tsp-\ell^2\nabla\lap\vc u\right).
\end{equation}
In this way we obtain a fluid which is quasi-Newtonian, while being
able to adhere to lower-dimensional objects. 

The only new parameter $\ell$ is given by $L/d$, hence it is
non-dimensional and strictly positive.  As shown
in~\cite[Sec.~2.1]{GiuFri12}, the drag force exerted on a point
particle moving in a hyperviscous fluid with hyperviscosity $\zeta=\mu
L^2$ is identical to the drag force exerted on a sphere of radius $L$
uniformly translating in a Newtonian fluid.
In~\cite[Sec.~3]{GiuMar10}, a similar result has been found for a pipe
flow driven by the motion of an inner cylinder with vanishing radius.
On the basis of those results, we can assign to $L$, and consequently
to $\ell$, the geometric meaning of an effective thickness of the
slender body $\bs$.

It is straightforward to check that $\vt T$, as defined
in~\eqref{eq:defT}, enjoys the standard symmetry and frame
indifference properties and satisfies a dissipation inequality. In
particular, one has
\begin{equation*}
\dvg \vt T = -\nabla p + \Delta \vc u -\ell^2\Delta\Delta \vc u.
\end{equation*}

Now we briefly discuss the functional setting. In view of the natural
variational formulation of the problem, we introduce the space
\[
C:=\{\vc u\in C^\infty_0(\R^3;\R^3):\ \dvg\vc u=0\}
\]
endowed with the norm
\[
\|\vc u\|^2:= \int_{\R^3}\left(2|\sym\nabla\vc
  u|^2+\ell^2|\lap\vc u|^2\right),
\]
where $\sym\nabla\vc u:=(\nabla\vc u+\nabla\vc u^\tsp)/2$.  Denote
with $\mathcal{X}$ the completion of $C$ in that norm; it is easy to
see that if $\vc u\in \mathcal{X}$, then $\nabla\vc u$ belongs to the
Sobolev space $W^{1,2}(\R^3;\R^9)$. Moreover, by the Sobolev Embedding
Theorem, $\mathcal{X}$ embeds in $L^q(\R^3;\R^3)$ for every $6\leq
q\leq \infty$ and also in a space of H\"older-continuous functions.
Regarding the pressure $p$ as the Lagrange multiplier of the
constraint $\dvg\vc u=0$, we will take it in the dual Sobolev space
$W^{-1,2}(\R^3)$.

We summarize the problem of the steady free fall of a one-dimensional
body $\bs$ at low Reynolds number, as the following: find $(\vc
u,p)\in \mathcal{X}\times W^{-1,2}(\R^3)$ and
$\vc\xi,\vc\omega,\vc g\in\R^3$ with $|\vc g|=1$, such that
\begin{align}
& \nabla p - \Delta \vc u +\ell^2\Delta\Delta \vc u = \vc g
\quad\text{on $\R^3\setminus \bs$},\label{eq:st2}\\
&\vc u(\vc x)=\vc\xi+\vc\omega\times\vc x\quad\text{on $\bs$},\label{eq:st3}\\
&m_e\vc g = -\vc f,\label{eq:st4}\\
&m_c\vc r\times\vc g = \vc t,\label{eq:st5}\\
&\vc g\times\vc\omega=\vc 0.\label{eq:st6}
\end{align}
The constraint $\dvg\vc u=0$ is encoded in the definition of the space
$\mathcal{X}$, while the strong decay condition~\eqref{eq:decay} is replaced by
an integrability condition for $\vc u$ on the whole $\R^3$. Notice
that, although equation~\eqref{eq:st2} is linear, the full problem is
nonlinear.

\section{Steady free fall at low Reynolds number}
\label{sec:freefall}

In this section we prove the existence of a solution for the the
differential problem~\eqref{eq:st2}--\eqref{eq:st6}. We begin by
introducing some auxiliary problems, which are well-posed by virtue of
the following result.

\begin{lem}\label{lem:aux}
Given $\vc\xi,\vc\omega\in\R^3$, there exists a unique
solution $(\vc h,p)\in \mathcal{X}\times W^{-1,2}(\R^3)$ of the problem
\begin{equation}\label{eq:aux}
\begin{cases}
\nabla p- \Delta \vc h +\ell^2\Delta\Delta \vc h=0 &\text{in
  $\R^3\setminus \bs$,}\\
\vc h=\vc\xi+\vc\omega\times\vc x &\text{on $\bs$.}
\end{cases}
\end{equation}
Moreover $\vc h\in W^{3,q}_{loc}(\R^3;\R^3)$ for any $1<q<\frac 3 2$.
\end{lem}
\begin{proof}
  Since $\mathcal{X}$ embeds in a space of H\"older-continuous functions, the
  subset
\[
\{\vc v\in \mathcal{X}:\ \vc v=\vc\xi+\vc\omega\times\vc x\text{ on }\bs\}
\]
is well-defined, closed and convex. The velocity field $\vc h$ can be
found by minimizing on that set the functional
\[
\mathcal F(\vc v):=\frac{1}{2}\|\vc v\|^2=\frac 1 2 \int_{\R^3}\left(2|\sym\nabla\vc
  v|^2+\ell^2|\lap\vc v|^2\right).
\]
Being $\mathcal F$ a strictly convex functional, $\vc h$ is
unique. Then, the pressure field $p$ can be recovered as the Lagrange
multiplier of the divergence-free constraint.

Since the adherence condition on $\bs$ can be replaced by a non
homogeneous right-hand side which is a measure supported on $\bs$,
that is
\[
\nabla p- \Delta \vc h +\ell^2\Delta\Delta \vc h=\vc \eta \quad\text{in $\R^3$,}\\
\]
where $\vc\eta$ is a (vector-valued) Radon measure which vanishes
outside $\bs$, we get for $\vc h$ a fourth-order linear elliptic
equation with a measure-valued datum. The space of
Radon measures embeds in $W^{-1,q}_{loc}(\R^3;\R^3)$ for
every $1<q<\frac 3 2$ (here $\frac 3 2$ is such that $q'>n$ in the
case $n=3$), hence a standard regularity gain of the solution
\cite[Theorem $15.3'$]{AgmDou59} entails $\vc h\in
W^{3,q}_{loc}(\R^3;\R^3)$.
\end{proof}

Consider now the solutions $(\vc h^{(i)},p^{(i)})$ and $(\vc
H^{(i)},P^{(i)})$ ($i=1,2,3$) in the space $\mathcal{X}\times W^{-1,2}(\R^3)$ of
the auxiliary problems
\begin{equation}
\label{eq:aux1}
\begin{cases}
\nabla p^{(i)}- \Delta \vc h^{(i)} +\ell^2\Delta\Delta\vc h^{(i)}=0 &\text{in }\R^3\setminus\bs,\\
\vc h^{(i)}=\vc e_i &\text{on }\bs,
\end{cases}
\end{equation}
and
\begin{equation}
\label{eq:aux2}
\begin{cases}
\nabla P^{(i)}- \Delta \vc H^{(i)} +\ell^2\Delta\Delta \vc H^{(i)}=0 &\text{in }\R^3\setminus\bs,\\
\vc H^{(i)}=\vc e_i\times\vc x &\text{on }\bs.
\end{cases}
\end{equation}
We will show that the combinations
\begin{equation}\label{eq:combination}
\vc u=\sum_{i=1}^3 [\xi_i\vc h^{(i)}+\omega_i\vc H^{(i)}],\quad 
p=\sum_{i=1}^3 [\xi_i p^{(i)}+\omega_i P^{(i)}]+\vc g\cdot\vc x,
\end{equation}
for a suitable choice of the vectors $\vc\xi$ and $\vc\omega$, solve
the steady free fall problem. First we need to introduce four matrices,
which will be closely studied in Sections~\ref{sec:res}
and~\ref{sec:sym}.

\begin{defn}
The matrices $\vt K$, $\vt S$, $\vt C$, and $\vt B$
are defined in Cartesian components by
\begin{gather}
\label{eq:tensorK}
\vt K_{ji}:=-\lim_{r\to 0}\int_{\de V_r(\bs)}\vt T(\vc
h^{(i)},p^{(i)})\vc n\cdot\vc e_j,\\
\label{eq:tensorS}
\vt S_{ji}:=-\lim_{r\to 0}\int_{\de V_r(\bs)}\vt T(\vc
H^{(j)},P^{(j)})\vc n\cdot\vc e_i,\\
\label{eq:tensorC}
\vt C_{ji}:=-\lim_{r\to 0}\int_{\de V_r(\bs)}\vc x\times\vt T(\vc
h^{(j)},p^{(j)})\vc n\cdot\vc e_i,\\
\label{eq:tensorB}
\vt B_{ji}:=-\lim_{r\to 0}\int_{\de V_r(\bs)}\vc x\times\vt T(\vc
H^{(i)},P^{(i)})\vc n\cdot\vc e_j,
\end{gather}
where $\vc n$ is the outer normal to
$V_r(\bs)$. Following~\cite{Bre64}, they are called {\em resistance
  tensors}, and in particular $\vt K$ is the {\em translation tensor},
$\vt B$ the {\em rotation tensor}, and $\vt S$ and $\vt C$ the
{\em coupling tensors}. Moreover, we denote by $\vt A$ the $6\times 6$ matrix
\[
\vt A:=
\begin{pmatrix}
\vt K & \vt S \\
\vt C & \vt B
\end{pmatrix}.
\]
\end{defn}
We postpone to Theorems~\ref{thm:res} and~\ref{thm:pos}, in the next
section, the proof of a fundamental property: 
\begin{center}
\emph{the matrices $\vt K$,
$\vt B$ and $\vt A$ are symmetric and positive definite.}
\end{center}
Although an energetic argument of Brenner~\cite[Section
5--2]{HapBre65} is usually adopted in this case, we will prefer to
give a direct proof.

Now we can prove the main theorem of the section.
\begin{thm}[Existence Theorem]\label{thm:sol}
The differential problem \eqref{eq:st2}--\eqref{eq:st6} admits a solution $(\vc u,p,\vc\xi,\vc\omega,\vc g)$.
\end{thm}
\begin{proof}
  It is straightforward to check that the fields $\vc u$ and $p$
  defined by~\eqref{eq:combination} satisfy equations
  \eqref{eq:st2} and \eqref{eq:st3}. Equation \eqref{eq:st6} implies
  that $\vc\omega=\lambda\vc g$ for some $\lambda\in\R$, and equations
  \eqref{eq:st4} and \eqref{eq:st5} reduce to the following algebraic
  system in the six scalar unknowns $\vc\xi$, $\lambda$, and $\vc g$
  (recall that $|\vc g|=1$):
\begin{equation}\label{eq:algsys}
\left\{
\begin{aligned}
\vt K\vc\xi+\lambda\vt S\vc g&\mbox{}=m_e\vc g\\
\vt C\vc\xi+\lambda\vt B\vc g&\mbox{}=-m_c\vc r\times\vc g.
\end{aligned}
\right.
\end{equation}

It is now clear that the steady free fall problem admits a solution if
and only if~\eqref{eq:algsys} admits a solution,
and the latter fact is related to the properties of the matrix
\[
\vt A=
\begin{pmatrix}
\vt K & \vt S \\
\vt C & \vt B
\end{pmatrix}.
\]
Since $\vt K$ is non singular, the first equation
of~\eqref{eq:algsys} becomes
\begin{equation*}
\vc\xi =\vt K^{-1}(m_e\vc g-\lambda\vt S\vc g)
\end{equation*}
and one can eliminate $\vc\xi$ in the second equation
of~\eqref{eq:algsys}. Since $\vt A$ is non-singular, the linear
transformation
\[
\vt F\vc g:=(\vt C\vt K^{-1}\vt S-\vt B)^{-1}(m_e\vt C\vt K^{-1}\vc
g+m_c\vc r\times\vc g)
\]
is well-defined and non-singular, and we can
write~\eqref{eq:algsys} as
\begin{equation}
\label{eq:finsys}
\left\{
\begin{aligned}
\vc\xi & =\vt K^{-1}(m_e\vc g+\lambda\vt S\vc g)\\
\vt F\vc g & =\lambda\vc g.
\end{aligned}
\right.
\end{equation}
Being $\vt F$ a $3\times 3$ real matrix, it has at least one
real eigenvalue. Such an eigenvalue $\lambda$, the associated unit
eigenvector $\vc g$ and $\vc\xi$ calculated as in the first equation
of \eqref{eq:finsys}, together with the fields $\vc u$ and $p$
introduced in~\eqref{eq:combination}, furnish a solution for equations
\eqref{eq:st2}--\eqref{eq:st6}.
\end{proof}

\section{An analysis of the resistance tensors}\label{sec:res}

In~\eqref{eq:tensorK}--\eqref{eq:tensorB} we introduced the four
resistance tensors $\vt K$, $\vt S$, $\vt C$ and $\vt B$. In view of
the conditions on $\bs$ assumed in the auxiliary
problems~\eqref{eq:aux1}--\eqref{eq:aux2}, we can give an equivalent
characterization.
\begin{prop}\label{prop:tensors}
The resistance tensors are such that
\begin{gather*}
\vt K_{ji}=-\lim_{r\to 0}\int_{\de V_r(\bs)}\vt T(\vc
h^{(i)},p^{(i)})\vc n\cdot\vc h^{(j)},\\
\vt S_{ji}=-\lim_{r\to 0}\int_{\de V_r(\bs)}\vt T(\vc
H^{(j)},P^{(j)})\vc n\cdot\vc h^{(i)},\\
\vt C_{ji}=-\lim_{r\to 0}\int_{\de V_r(\bs)}\vt T(\vc
h^{(j)},p^{(j)})\vc n\cdot\vc H^{(i)},\\
\vt B_{ji}=-\lim_{r\to 0}\int_{\de V_r(\bs)}\vt T(\vc
H^{(i)},P^{(i)})\vc n\cdot\vc H^{(j)}.
\end{gather*}
\end{prop}

\begin{proof}
Since $\vc h^{(i)}$ is continuous and $\vc h^{(i)}=\vc e_i$ on $\bs$,
one has, as $r\to 0$,
\[
\|\vc h^{(i)}-\vc e_i\|_{\infty,\de V_r(\bs)}:=\sup_{\vc x \in\de
  V_r(\bs)}|\vc h^{(i)}(\vc x)-\vc e_i|\to 0.
\]
Hence, considering for instance the translation tensor $\vt K$, it
follows that
\begin{multline*}
\left|\int_{\de V_r(\bs)}\vt T(\vc h^{(i)},p^{(i)})\vc
  n\cdot\vc h^{(j)}-
\int_{\de V_r(\bs)}\vt T(\vc h^{(i)},p^{(i)})\vc n\cdot\vc e_j\right|\\
\leq \|\vc h^{(j)}-\vc e_j\|_{\infty,\de V_r(\bs)}
\int_{\de V_r(\bs)}\left|\vt T(\vc h^{(i)},p^{(i)})\vc n\right|\to 0
\end{multline*}
as $r\to 0$. The proof of the remaining three formulae is similar.
\end{proof}

In the remainder of the section we will prove that $\vt K$ and $\vt B$
are symmetric and that $\vt C^\tsp=\vt S$. We first need a fundamental
property of steady incompressible flows at low Reynolds number, the
so-called {\em Reciprocal Theorem} (see \cite[Sec.\ 3-5]{HapBre65}),
which, roughly speaking, states a reciprocity property between two
solutions of the same equation, independently of the boundary
conditions. The validity of the theorem, which is quite trivial for
ordinary fluids, is not so obvious in the present case of hyperviscous
fluids, since the lack of further boundary conditions and the higher order
of the differential operator can break such a reciprocity. However, the
theorem can be recovered for the particular case of one-dimensional
bodies.

\begin{thm}[Reciprocal Theorem]\label{thm:reciprocal}
  Let $(\vc u_1,p_1),(\vc u_2,p_2)$ be two solutions in the space $\in
  \mathcal{X}\times W^{-1,2}(\R^3;\R)$ of the equation
\begin{equation}
\dvg\vt T(\vc u,p)=0\quad\text{in $\R^3\setminus\bs$,}
\end{equation}
where $\vt T(\vc u,p)$ is defined as in~\eqref{eq:defT}. Assume that
$\vc u_1,\vc u_2\in W^{3,\frac 6 5}_{loc}(\R^3;\R^3)$.
Then we have
\[
\lim_{r\to 0}\int_{\partial V_r(\bs)}\vt T(\vc
u_1,p_1)\vc n\cdot\vc u_2 = 
\lim_{r\to 0}\int_{\partial V_r(\bs)}\vt T(\vc
u_2,p_2)\vc n\cdot\vc u_1. 
\]
\end{thm}

\begin{proof}
  Consider a large ball $B_R$ containing $V_r(\bs)$ and apply
  Gauss-Green formula to the domain $B_R\setminus V_r(\bs)$: then
\[
\int_{\partial V_r(\bs)}\vt T(\vc u_1,p_1)\vc n\cdot\vc u_2 = 
-\int_{B_R\setminus V_r(\bs)}\vt T(\vc u_1,p_1)\cdot\nabla\vc u_2-
  \int_{\partial B_R}\vt T(\vc u_1,p_1)\vc n\cdot \vc u_2,
\]
where the normal in the left-hand side is exterior to $V_r(\bs)$ and
we kept into account that $\dvg\vt T(\vc u_2,p_2)=0$ on $\R^3\setminus
\bs$. The last surface integral on $\partial B_R$ vanishes as
$R\to+\infty$, since any solution to the hyperviscous Stokes' problem
decays as $1/|\vc x|$ (see Appendix~\ref{app}), hence $\vt T$ decays
as $1/|\vc x|^2$.

Now take the first term of the right-hand
side and use the constitutive prescription~\eqref{eq:defT}:
\begin{multline*}
-\int_{B_R\setminus V_r(\bs)}\vt T(\vc u_1,p_1)\cdot\nabla\vc u_2\\
=
-\int_{B_R\setminus V_r(\bs)}(\nabla\vc u_1+\nabla\vc u_1^\tsp)\cdot\nabla\vc u_2
+\ell^2 \int_{B_R\setminus V_r(\bs)}\nabla\lap\vc u_1\cdot \nabla\vc u_2\\ 
=-\int_{B_R\setminus V_r(\bs)}(\nabla\vc u_2+\nabla\vc u_2^\tsp)\cdot\nabla\vc u_1
+\ell^2 \int_{B_R\setminus V_r(\bs)}\nabla\lap\vc u_1\cdot \nabla\vc
u_2.
\end{multline*}
Consider the last term; since the gradient and Laplace operators
commute, by Green's second identity it follows that
\begin{align*}
\int_{B_R\setminus V_r(\bs)}\nabla\lap\vc u_1\cdot \nabla\vc
u_2 = &
\int_{B_R\setminus V_r(\bs)}\nabla\vc u_1\cdot \nabla\lap\vc
u_2\\
- & \int_{\de {V_r(\bs)}}\left[\nabla\vc u_2\cdot(\nabla\nabla\vc
   u_1)\vc n-
\nabla\vc u_1\cdot(\nabla\nabla\vc u_2)\vc n\right]\\
+ & \int_{\de B_R}\left[\nabla\vc u_2\cdot(\nabla\nabla\vc
   u_1)\vc n- 
\nabla\vc u_1\cdot(\nabla\nabla\vc u_2)\vc n\right].
\end{align*}
Now we claim that the surface integrals vanish as $r\to 0$ and $R\to
+\infty$. Indeed, take for instance the term
\[
\int_{\de V_r(\bs)}\nabla\vc u_2\cdot(\nabla\nabla\vc u_1)\vc n
\]
and apply again Gauss-Green formula inside $V_r(\bs)$. Then
\[
\int_{\de V_r(\bs)}\nabla\vc u_2\cdot(\nabla\nabla\vc u_1)\vc n = 
\int_{V_r(\bs)}[\nabla\nabla\vc u_1 \cdot \nabla\nabla\vc u_2 + 
\nabla\vc u_2\cdot\lap\nabla\vc u_1].
\]
Since $\vc u_1,\vc u_2\in \mathcal{X}$, then  $\nabla\nabla\vc
u_1\cdot\nabla\nabla\vc u_2 \in L^1(\R^3)$ and
\[
\lim_{r\to 0}\int_{V_r(\bs)}\nabla\nabla\vc u_1 \cdot
\nabla\nabla\vc u_2  =0.
\]
Moreover, since $\nabla\vc u_2\in
W^{1,2}(\R^3;\R^9)$, the Sobolev Embedding Theorem
ensures that $\nabla\vc u_2 \in
L^6(V_r(\bs);\R^9)$. We also assumed that
$\vc u_1,\vc u_2\in W^{3,\frac 6 5}_{loc}(\R^3;\R^3)$,
hence $\lap\nabla\vc u_1\in L^\frac 6 5(V_r(\bs);\R^9)$, $\nabla\vc
u_2\cdot\lap\nabla\vc u_1\in L^1(V_r(\bs))$ and we can conclude that
\[
\lim_{r\to 0}\int_{V_r(\bs)} \nabla\vc u_2\cdot\lap\nabla\vc u_1=0.
\]
Regarding the surface integral on $\partial B_R$, it vanishes as
$R\to+\infty$ since $\vc u_1,\vc u_2$ decay as $1/|\vc x|$ (see Appendix~\ref{app}). 
\end{proof}

Now we can easily obtain the main result of the section.
\begin{thm}\label{thm:res}
The resistance tensors are such that
\[
\vt K^\tsp=\vt K,\quad\vt B^\tsp=\vt B,\quad \vt C^\tsp=\vt S.
\]
In particular, the matrix $\vt A$ is symmetric.
\end{thm}

\begin{proof}
  Consider the solutions $(\vc h^{(i)},p^{(i)})$ and $(\vc
  H^{(i)},P^{(i)})$ of the auxiliary
  problems~\eqref{eq:aux1}-\eqref{eq:aux2}. By Lemma~\ref{lem:aux} it
  follows that $\vc h^{(i)},\vc H^{(i)}\in W^{3,q}_{loc}(\R^3;\R^3)$
  for any $1<q<\frac 3 2$, in particular for $q=\frac 65$. Then, the
  Reciprocal Theorem applies to such solutions and, by combining
  it with Proposition~\ref{prop:tensors}, we conclude the proof.
\end{proof}

We are now in a position to prove the positive definiteness
of the tensors $\vt K$, $\vt B$ and $\vt A$.
\begin{thm}\label{thm:pos}
The $6\times 6$ matrix $\vt A$ is positive definite. As a consequence,
also $\vt K$ and $\vt B$ are positive definite.
\end{thm}

\begin{proof}
  For $\vc\xi=(\xi_1,\xi_2,\xi_3)$ and
  $\vc\omega=(\omega_1,\omega_2,\omega_3)$ set 
\[
\vc u =\sum_{i=1}^3[\xi_i\vc h^{(i)}+\omega_i\vc H^{(i)}],\quad 
p=\sum_{i=1}^3[\xi_ip^{(i)}+\omega_iP^{(i)}],
\]
where $(\vc h^{(i)},p^{(i)})$ and $(\vc H^{(i)},P^{(i)})$ are the
solutions of the auxiliary problems~\eqref{eq:aux1}--\eqref{eq:aux2}.
Using Proposition~\ref{prop:tensors}, the Reciprocal Theorem and the
linearity of $\vt T$ one can check that
\begin{gather*}
\begin{pmatrix}\vc\xi\\ \vc\omega\end{pmatrix}\cdot
\vt A \begin{pmatrix}\vc\xi\\ \vc\omega\end{pmatrix} =\vc\xi\cdot \vt
K\vc\xi + 2\vc\omega\cdot \vt C\vc\xi +
\vc\omega\cdot \vt B\vc\omega
=-\lim_{r\to 0} 
\int_{\de V_r(\bs)}\vt T(\vc u,p)\vc n\cdot\vc u.
\end{gather*}
Now argue as in the proof of the Reciprocal Theorem. Take a large ball
$B_R$ containing $V_r(\bs)$ and apply Gauss-Green formula to the
domain $B_R\setminus V_r(\bs)$ to obtain
\[
-\int_{\partial V_r(\bs)}\vt T(\vc u,p)\vc n\cdot\vc u = 
\int_{B_R\setminus V_r(\bs)}\vt T(\vc u,p)\cdot\nabla\vc u+
  \int_{\partial B_R}\vt T(\vc u,p)\vc n\cdot \vc u,
\]
where the normal in the left-hand side is exterior to $V_r(\bs)$ and
we kept into account that $\dvg\vt T(\vc u,p)=0$ on $\R^3\setminus
\bs$. The last surface integral on $\partial B_R$ vanishes as
$R\to+\infty$ since the solution decays as
$1/|\vc x|$, hence $\vt T$ decays as $1/|\vc x|^2$.

Now consider the first term of the right-hand side and use the
constitutive prescription~\eqref{eq:defT}:
\begin{multline}\label{p}
\int_{B_R\setminus V_r(\bs)}\vt T(\vc u,p)\cdot\nabla\vc u
=\int_{B_R\setminus V_r(\bs)}(\nabla\vc u+\nabla\vc u^\tsp)\cdot\nabla\vc u
-\ell^2 \int_{B_R\setminus V_r(\bs)}\nabla\lap\vc u\cdot \nabla\vc u\\ 
=\int_{B_R\setminus V_r(\bs)}|\nabla\vc u|^2
+\int_{B_R\setminus V_r(\bs)}\nabla\vc u^\tsp\cdot\nabla\vc u
-\ell^2 \int_{B_R\setminus V_r(\bs)}\nabla\lap\vc u\cdot \nabla\vc u.
\end{multline}
We deal with the second integral, taking into account that $\dvg\vc u=0$:
\[
\int_{B_R\setminus V_r(\bs)}\nabla\vc u^\tsp\cdot\nabla\vc u
= -\int_{\de {V_r(\bs)}}\vc u\cdot(\nabla\vc u^\tsp \vc n)
+\int_{\de B_R}\vc u\cdot(\nabla\vc u^\tsp \vc n).
\]
The last integral on $\de B_R$ vanishes as $R\to\infty$ for the usual
asymptotic behavior at infinity. By applying the Gauss-Green theorem
on the other integral it follows that
\[
- \int_{\de {V_r(\bs)}}\vc u\cdot(\nabla\vc u^\tsp \vc n)
=- \int_{V_r(\bs)}\nabla\vc u\cdot\nabla\vc u^\tsp -
 \int_{V_r(\bs)}\vc u\cdot\dvg(\nabla\vc u^\tsp)
\]
and both terms vanish as $r\to 0$ since the first integrand is in
$L^1$ and in the last integral one has $\dvg\vc u=0$.

Now we deal with the last integral of ~\eqref{p}.  Since the gradient
and Laplace operators commute, we have
\begin{align*}
-\ell^2\int_{B_R\setminus V_r(\bs)}\nabla\lap\vc u\cdot \nabla\vc
u & = \ell^2\int_{B_R\setminus V_r(\bs)}|\nabla\nabla\vc u|^2\\
& + \ell^2 \int_{\de {V_r(\bs)}}\nabla\vc u\cdot(\nabla\nabla\vc
   u)\vc n
- \ell^2\int_{\de B_R}\nabla\vc u\cdot(\nabla\nabla\vc
   u)\vc n.
\end{align*}
Following the last part of the proof of the Reciprocal Theorem, it can
be proved that the surface integrals vanish as $r\to 0$ and $R\to
+\infty$. Summarizing,
\[
\begin{pmatrix}\vc\xi\\ \vc\omega\end{pmatrix}\cdot
\vt A \begin{pmatrix}\vc\xi\\ \vc\omega\end{pmatrix}
=\lim_{\substack{r\to 0\\ R\to+\infty}}\int_{B_R\setminus V_r(\bs)}|\nabla\vc u|^2 + 
\ell^2\int_{B_R\setminus V_r(\bs)}|\nabla\nabla\vc u|^2,
\]
hence $\vt A$ is (strictly) positive definite.
\end{proof}

\section{Translational solutions for bodies with symmetries}
\label{sec:sym}

The free fall of a one-dimensional body in a hyperviscous fluid at low
Reynolds number is characterized by 21 independent coefficients: 12
coefficients for the tensors $\vt K$ and $\vt B$ and 9 coefficients
for the coupling tensor $\vt C$. However, material symmetries of the
body can significantly reduce such a number. Moreover, the symmetries
induce some restrictions on the form of the resistance tensors. We are
specifically interested in symmetries which induce purely
translational motions of the body (that is, with $\vc \omega=0$). We will
now study some particular symmetries.

If the body is invariant under a change of frame given by an
orthogonal matrix $\vt Q$, then also the solutions $\vc h^{(i)}$ of
the auxiliary problems~\eqref{eq:aux1} do not change; on the contrary,
the solutions $\vc H^{(i)}$ of~\eqref{eq:aux2}
undergo a sign change if $\det \vt Q=-1$, due to the presence of the
vector product in the boundary condition. Hence one can prove that
\begin{equation}
\label{eq:change}
\vt K=\vt Q^\tsp \vt K \vt Q,\quad 
\vt B=\vt Q^\tsp \vt B \vt Q,\quad
\vt C=(\det\vt Q)\vt Q^\tsp \vt C \vt Q.
\end{equation}

\subsection{Bodies with a plane of symmetry}
We say that the body $\bs$ has $x_2x_3$ as a {\em plane of
  material symmetry}, if the density function $\rho$ of the body satisfies
\[
\rho(-x_1,x_2,x_3)=\rho(x_1,x_2,x_3)\qquad\text{for every
  $(x_1,x_2,x_3)\in\R^3$}.
\]
In particular, a {\em homogeneous} body has a plane of material symmetry if,
and only if, it is symmetric with respect to that plane.
\begin{prop}\label{prop:planesym}
  Assume that $\bs$ has $x_2x_3$ as a plane of material symmetry. Then
  the resistance tensors have the form
\[
\vt K=
\begin{pmatrix}
\vt K_{11} & 0 & 0\\
0 & \vt K_{22} & \vt K_{23}\\
0 & \vt K_{23} & \vt K_{33}
\end{pmatrix},\quad
\vt B=
\begin{pmatrix}
\vt B_{11} & 0 & 0\\
0 & \vt B_{22} & \vt B_{23}\\
0 & \vt B_{23} & \vt B_{33}
\end{pmatrix},\quad
\vt C=
\begin{pmatrix}
0 & \vt C_{12} & \vt C_{13}\\
\vt C_{21} & 0 & 0\\
\vt C_{31} & 0 & 0
\end{pmatrix}.
\]
\end{prop}
\begin{proof}
  Since $\Sigma$ is invariant under the orthogonal transformation
  given by
\[
\vt Q=
\begin{pmatrix}
-1 & 0 & 0\\
0 & 1 & 0\\
0 & 0 & 1
\end{pmatrix},\quad
\det \vt Q = -1,
\]
then formulae~\eqref{eq:change} yield
\begin{gather*}
\vt K_{12}=\vt K_{13}=\vt B_{12}=\vt B_{13}=0,\\
\vt C_{11}=\vt C_{22}=\vt C_{33}=\vt C_{23}=\vt C_{32}=0.\qedhere
\end{gather*}
\end{proof}

Now consider the system~\eqref{eq:finsys} which solves the problem of
the steady free fall, in the case when the body $\Sigma$ has a
plane of material symmetry, say $x_2x_3$. Suppose moreover that
$\Sigma$ is homogeneous, so that the center of mass and the centroid
coincide, hence $\vc r=0$. In such a case the second equation of~\eqref{eq:finsys}
becomes
\begin{equation}
\label{eq:eigen}
m_e(\vt C\vt K^{-1}\vt C^\tsp-\vt B)^{-1}\vt C\vt K^{-1}\vc
g=\lambda\vc g.
\end{equation}
Since $\vc \omega=\lambda \vc g$, we get a translational
solution whenever $\lambda=0$. Being $(\vt C\vt K^{-1}\vt C^\tsp-\vt B)$
and $\vt K$ positive definite matrices,~\eqref{eq:eigen} has a
solution $\lambda=0$ if, and only if, $\det \vt C=0$. In the case of a
body with a plane of material symmetry, indeed, the latter condition
is satisfied and it is easy to check that an eigenvector of $\vt C$,
say $\vt u_0$, corresponding to the eigenvalue $\lambda=0$ lies in the
plane $x_2x_3$. Hence one has
\[
m_e(\vt C\vt K^{-1}\vt C^\tsp-\vt B)^{-1}\vt C\vt K^{-1}\vc g=0
\iff
\vc g=\vt K\vc u_0;
\]
by the form of $\vt K$ given in Proposition~\ref{prop:planesym}, also the
vector $\vt K \vc u_0$ lies in the plane $x_2x_3$. We can summarize
the latter result in the following theorem:
\begin{thm}
Assume that $\bs$ has $x_2x_3$ as a plane of material symmetry. Then
there exist an orientation of the body, lying in the same plane of
symmetry, which gives rise to a purely translational solution.
\end{thm}

Now it is quite easy to study the class of bodies with two orthogonal
planes of symmetry:
\begin{cor}
If the body has two orthogonal planes of symmetry, say $x_1x_3$ and
$x_2x_3$, then the free fall along the $x_3$-direction gives rise to a
 purely translational motion.
\end{cor}

\begin{proof}
Since the body is invariant under the orthogonal matrices
\[
\begin{pmatrix}
1 & 0 & 0\\
0 & -1 & 0\\
0 & 0 & 1
\end{pmatrix},\quad
\begin{pmatrix}
-1 & 0 & 0\\
0 & 1 & 0\\
0 & 0 & 1
\end{pmatrix},
\]
it is easy to check that $\vt K$ and $\vt B$ are diagonal, and $\vt C$
has the form
\[
\vt C=
\begin{pmatrix}
0 & \vt C_{12} & 0\\
\vt C_{21} & 0 & 0\\
0 & 0 & 0
\end{pmatrix}.
\]
Hence $\vc u_0=(0,0,a)$ is an eigenvector of $\vt C$ corresponding to the null
eigenvalue, and the motion with orientation given by 
\[
\vc g=\frac{\vt K\vc u_0}{|\vt K\vc u_0|}=(0,0,\pm 1)
\]
furnishes a purely translational solution.
\end{proof}

\subsection{Helicoidally symmetric bodies}

Now we study bodies which are invariant under the action of a rotation
of angle $\theta\in[0,2\pi[$ around the $x_1$-axis, which is represented
by the orthogonal matrix
\[
\vt R_\theta:=\begin{pmatrix}
1 & 0 & 0\\
0 & \cos\theta & -\sin\theta\\
0 & \sin\theta & \cos\theta
\end{pmatrix}.
\]
Following~\cite{HapBre65}, we say that a (one-dimensional) body $\bs$
is {\em helicoidally symmetric} if there exists a co-moving frame such that
\[
\vt R_\theta \bs = \bs\quad\text{for some $\theta\neq 0,\pi$},
\]
that is, if it is invariant under a discrete group of co-axial rotations of
order strictly greater than 2. For instance, a homogeneous body
composed of three concurrent edges of a regular tetrahedron is
helicoidally symmetric with $\theta=2\pi/3$.
\begin{prop}\label{prop:heli}
  Assume that $\bs$ is helicoidally symmetric around $x_1$. Then $\vt
  K$ and $\vt B$ are diagonal with $\vt K_{22}= \vt K_{33}$ and
  $\vt B_{22}= \vt B_{33}$, and $\vt C$ is of the form
\[
\vt C=
\begin{pmatrix}
\vt C_{11} & 0 & 0\\
0 & \vt C_{22} & \vt C_{23}\\
0 & -\vt C_{23} & \vt C_{33}
\end{pmatrix}.
\]
\end{prop}
\begin{proof}
  Let us employ formulae~\eqref{eq:change} with $\vt Q=\vt R_\theta$,
  keeping into account that $\det\vt R_\theta=1$. For the matrix $\vt
  C$ we get the conditions
\begin{align*}
\vt C_{12} & = \vt C_{12}\cos\theta + \vt C_{13}\sin\theta,\\
\vt C_{13} & = \vt C_{13}\cos\theta - \vt C_{12}\sin\theta,\\ 
\vt C_{21} & = \vt C_{21}\cos\theta + \vt C_{31}\sin\theta,\\
\vt C_{22} & = \vt C_{22}\cos^2\theta +(\vt C_{23}+\vt C_{32})\cos\theta\sin\theta
+\vt C_{33} \sin^2\theta,\\
\vt C_{23} & = \vt C_{23}\cos^2\theta +(\vt C_{33}-\vt C_{22})\cos\theta\sin\theta 
-\vt C_{32} \sin^2\theta,\\
\vt C_{31} & = \vt C_{31}\cos\theta -\vt C_{21}\sin\theta,\\
\vt C_{32} & = \vt C_{32}\cos^2\theta + (\vt C_{33}-\vt C_{22}) \cos\theta\sin\theta 
-\vt C_{23}\sin^2\theta,\\
\vt C_{33} & = \vt C_{33}\cos^2\theta - (\vt C_{23}+\vt C_{32}) \cos\theta\sin\theta 
+\vt C_{22}\sin^2\theta,
\end{align*}
which in turn imply that $\vt C_{12}=\vt C_{21}=\vt C_{13}=\vt
C_{31}=0$ and $\vt C_{23}+\vt C_{32}=0$, since $\theta\neq
0,\pi$. Being $\vt K$ and $\vt B$ symmetric, we have the
further conditions $\vt K_{23}=\vt B_{23}=0$ and $\vt K_{22}=\vt
K_{33}$, $\vt B_{22}=\vt B_{33}$.
\end{proof}

\subsection{Helicoidally symmetric bodies with fore-aft symmetry}

A remarkable situation is the case of a homogeneous one-dimensional
helicoidally symmetric body with {\em fore-aft symmetry}, that is, a
body which is both helicoidally symmetric around an axis, and has a
plane of symmetry orthogonal to that axis. A simple example is given
by a body composed of the 12 edges of a regular octahedron.

Without loss of generality, let us assume that a one-dimensional
body $\bs$ is helicoidally symmetric around $x_1$ and has $x_2x_3$ as
a plane of symmetry. Since the coupling tensor $\vt C$ has to satisfy
both Proposition~\ref{prop:planesym} and Proposition~\ref{prop:heli} at the same time,
it follows that
\[
\vt C = \vt 0.
\]
Assuming that $\Sigma$ be homogeneous, so that $\vc r=\vc 0$, the
system~\eqref{eq:finsys} merely becomes
\[
\begin{cases}
\vc\xi =m_e\vt K^{-1}\vc g,\\
\vc 0=\lambda\vc g,
\end{cases}
\]
hence $\lambda = 0$ for any direction $\vc g$. Then $\vc\omega=\vc 0$
and for any given orientation the body falls with a purely
translational velocity given by $\vc\xi =m_e\vt K^{-1}\vc g$.

\appendix

\section{Green's function for Stokes flow}
\label{app}

The basic tool used to construct solutions to the Stokes problem is
the so-called Stokeslet, that is the Green's function for the Stokes
operator in $\R^3$. In this Appendix\footnote{The results of the
  Appendix are based on~\cite{GiuFri12}.}  we want to compute the
expression of the Stokeslet in the case of our hyperviscous fluid,
identified by the operator
\[
\mathcal A:=\ell^2\lap\lap -\lap.
\]
We first need a Green's function $g$ solution of the fourth-order
elliptic equation
\[
\ell^2\lap\lap g - \lap g=\delta(\vc x-\vc x').
\]
Using the Fourier transform, we easily obtain
\[
g(\vc x-\vc x')=\frac{1}{(2\pi)^3}\int_{\R^3} \frac{e^{i\vc k\cdot(\vc x-\vc x')}}{ |\vc k|^2(\ell^2|\vc k|^2+1)}d\vc k.
\]
We choose a basis for the momentum space in such a way that $\vc x-\vc x'$ is along the $k_3$-direction, set $R=|\vc x-\vc x'|$, switch to polar coordinates $(k,\theta,\phi)$, and use the calculus of residues to obtain
\begin{align*}
g(\vc x-\vc x')&\mbox{}=\frac{2\pi}{(2\pi)^3\ell^2}\int_0^{+\infty}\int_{-1}^1\frac{e^{ikR\cos\theta}}{k^2+1/\ell^2}d(\cos\theta)dk\\
&\mbox{}=\frac{2}{(2\pi)^2\ell^2 R}\int_0^{+\infty}\frac{\sin{kR}}{k(k^2+1/\ell^2)}dk\\
&\mbox{}=\frac{1}{(2\pi)^2\ell^2
  R}\imp\left[\int_{-\infty}^{+\infty}\frac{e^{ikR}}{k(k^2+1/\ell^2)}dk\right]\\
&\mbox{}=\frac{1}{(2\pi)^2\ell^2 R}\left(\pi \ell^2-\pi \ell^2e^{-\frac{R}{\ell}}\right).
\end{align*}
Hence, the Green's function is
\begin{equation}\label{eq:g}
g(\vc x-\vc x')=\frac{1}{4\pi |\vc x-\vc x'|}\left[1-\exp\left(-\frac{|\vc x-\vc x'|}{\ell}\right)\right].
\end{equation}
Notice that, in the limit $\ell\to 0$, \eqref{eq:g} reduces to the fundamental solution
\begin{equation}\label{eq:g1}
g_1(\vc x-\vc x')=\frac{1}{4\pi |\vc x-\vc x'|}
\end{equation}
for the Laplace operator. Moreover, $g$ is well defined for any $\vc
x\in\R^3$, at variance with the classical expression $g_1$, which is
singular at the origin.

We now proceed to construct the {\em hyperviscous Stokeslet}, that is a
pressure field $\sletp$ and a velocity field $\slet$ satisfying
\begin{equation}\label{eq:div0}
\dvg\slet=0,
\end{equation}
\begin{equation}\label{eq:stokes}
\nabla\sletp+\mathcal A\slet =\vc h\delta(\vc x),
\end{equation}
with $\vc h\in\R^3$ and $\mathcal A=\ell^2\lap\lap - \lap$.
Let $\phi$ satisfy $\mathcal A\phi=\delta(\vc x-\vc x')$; then, since
$\mathcal A$ commutes with $\nabla$, a solution
for~\eqref{eq:div0}--\eqref{eq:stokes} is given by
\[
\sletp=-\mathcal A\vartheta,\quad
\slet=\vc h\phi+\nabla\vartheta.
\]
The scalar field $\vartheta$ entering this solution is chosen to satisfy the constraint~\eqref{eq:div0} and turns out to have the explicit form
\[
\vartheta=(-\lap)^{-1}(\vc h\cdot\nabla \phi)=g_1*(\vc h\cdot\nabla\phi),
\]
where $g_1$ is as defined in \eqref{eq:g1} and $*$ denotes the usual
convolution product. Now, exploiting the properties of the
convolution and the operator $\mathcal A$, and using the Green's
function $g$ given by equation~\eqref{eq:g}, we find that
\[
-(g_1*\mathcal A(\vc h\cdot\nabla g))=-(g_1*\dvg(\mathcal A(g\vc h)))=-\dvg(g_1*\mathcal A(g\vc h))=-\vc h\cdot\nabla g_1
\]
and, denoting by $\ft{f}$ the Fourier transform of the function $f$,
\begin{align*}
g_1*(\vc h\cdot\nabla g)&\mbox{}=\frac{1}{(2\pi)^3}\int i(\vc
h\cdot\vc k)\ft{g_1}\ft{g}e^{i\vc k\cdot\vc x}d\vc k\\
&\mbox{}
=\frac{1}{(2\pi)^3}\int\frac{i(\vc h\cdot\vc k)e^{i\vc k\cdot\vc x}}{ |\vc k|^4(\ell^2|\vc k|^2+1)}d\vc k\\
&\mbox{}
=\frac{\vc h\cdot\vc x}{4\pi^2\ell^2|\vc x|}\int_0^{+\infty}\int_{-1}^{1}\frac{i\cos\theta e^{ik|\vc x|\cos\theta}}{k(k^2+1/\ell^2)}d(\cos\theta)dk\\
&\mbox{}=\frac{-\vc h\cdot\vc x}{4\pi^2\ell^2|\vc x|}\int_{-1}^{1}\tau\int_0^{+\infty}\frac{\sin(k|\vc x|\tau)}{k(k^2+1/\ell^2)} dk d\tau\\
&\mbox{}
=\frac{-\vc h\cdot\vc x}{8\pi |\vc x|}\int_{-1}^{1}|\tau|\left(1-e^{-\frac{|\vc x|}{\ell}|\tau|} \right)d\tau\\
&\mbox{}=-\frac{\vc h\cdot\vc x}{8\pi |\vc x|}\left[1+\frac{2\ell}{|\vc x|}e^{-\frac{|\vc x|}{\ell}}+\frac{2\ell^2}{|\vc x|^2}\left(e^{-\frac{|\vc x|}{\ell}}-1\right)\right]
.
\end{align*}
Hence the Stokeslet is given by
\begin{equation*}
\sletp(\vc x)=\frac{\vc h\cdot\vc x}{4\pi|\vc x|^3},
\end{equation*}
\begin{multline*}
\slet(\vc x)=\frac{\vc h}{8\pi |\vc x|}\left[1-2e^{-\frac{|\vc x|}{\ell}}-\frac{2\ell}{|\vc x|}e^{-\frac{|\vc x|}{\ell}}-\frac{2\ell^2}{|\vc x|^2}\left(e^{-\frac{|\vc x|}{\ell}}-1\right)\right]
\\
+\frac{(\vc h\cdot\vc x)\vc x}{8\pi |\vc x|^3}\left[1+2e^{-\frac{|\vc x|}{\ell}}+\frac{6\ell}{|\vc x|}e^{-\frac{|\vc x|}{\ell}}+\frac{6\ell^2}{|\vc x|^2}\left(e^{-\frac{|\vc x|}{\ell}}-1\right)\right]
.
\end{multline*}
We also define the hyperviscous Oseen tensor $\tslet$ as, using
Cartesian components,
\begin{multline*}
\tsletc_{ij}(\vc x):=\frac{\delta_{ij}}{8\pi |\vc x|}\left[1-2e^{-\frac{|\vc x|}{\ell}}-\frac{2\ell}{|\vc x|}e^{-\frac{|\vc x|}{\ell}}-\frac{2\ell^2}{|\vc x|^2}\left(e^{-\frac{|\vc x|}{\ell}}-1\right)\right]
\\
+\frac{x_ix_j}{8\pi |\vc x|^3}\left[1+2e^{-\frac{|\vc x|}{\ell}}+\frac{6\ell}{|\vc x|}e^{-\frac{|\vc x|}{\ell}}+\frac{6\ell^2}{|\vc x|^2}\left(e^{-\frac{|\vc x|}{\ell}}-1\right)\right]
,
\end{multline*}
whereby it follows that $\slet(\vc x)=\tslet(\vc x)\vc h$.
The Stokeslet allows us to obtain an integral representation for the
solution of
\begin{equation*}
\nabla p- \lap(\vc u-\ell^2\lap\vc u)=\rho\vc b,
\end{equation*}
with vanishing condition at infinity, in the form of a convolution:
\begin{equation}\label{eq:convol}
\vc u(\vc x):=\rho\int_{\R^3} \tslet(\vc x-\vc x')\vc b(\vc x')d\vc x'.
\end{equation}
In particular, whenever $\vc b$ has compact support, such as in the
case of the gravity force acting on a bounded body $\Sigma$, the
solution $\vc u$ behaves as $1/|\vc x|$ for $|\vc x|\to \infty$.

\section*{Acknowledgments} 
This research is partially supported by GNFM (Gruppo Nazionale per la Fisica Matematica).


\end{document}